\documentclass[conference]{IEEEtran}
\IEEEoverridecommandlockouts

\usepackage{cite}
\usepackage{ieeefig}
\usepackage{amsfonts}
\usepackage{epsfig}
\usepackage{graphicx}
\usepackage{stfloats}
\newtheorem{lemma}{Lemma}

\begin{document}
\title{A Large-System Analysis of the Imperfect-CSIT Gaussian Broadcast Channel
with a DPC-based Transmission Strategy}

\author{Chinmay S.~Vaze %
        and~Mahesh K.~Varanasi% <-this % stops a space
\thanks{This work was supported in part by NSF Grant CCF-0728955.
The authors are with the Department of Electrical, Computer, and
Energy Engineering, University of Colorado, Boulder, CO 80309-0425
USA (e-mail: {vaze, varanasi}@colorado.edu).} }

%\markboth{Journal of \LaTeX\ Class Files,~Vol.~6, No.~1, January~2007}%
%{Shell \MakeLowercase{\textit{et al.}}: Bare Demo of IEEEtran.cls for Journals}

\markboth{}{C. S. Vaze and M. K. Varanasi: Short Title}


\maketitle

\begin{abstract}
The Gaussian broadcast channel (GBC) with $K$ transmit antennas and
$K$ single-antenna users is considered for the case in which the
channel state information is obtained at the transmitter via a
finite-rate feedback link of capacity $r$ bits per user. The
throughput (i.e., the sum-rate normalized by $K$) of the GBC is
analyzed in the limit as $K \to \infty$ with $\frac{r}{K} \to
\bar{r}$. Considering the transmission strategy of zeroforcing dirty
paper coding (ZFDPC), a closed-form expression for the asymptotic
throughput is derived. It is observed that, even under the
finite-rate feedback setting, ZFDPC achieves a significantly higher
throughput than zeroforcing beamforming. Using the asymptotic
throughput expression, the problem of obtaining the number of users
to be selected in order to maximize the throughput is solved.
\end{abstract}


% \IEEEpeerreviewmaketitle

\section{Introduction}
\IEEEPARstart{T}{he} Gaussian broadcast channel (GBC) has been
intensely researched in recent
years. It is well-known that dirty paper coding (DPC) \cite{Costa}
achieves the capacity region of the GBC if perfect channel state
information (CSI) is available at the transmitter (CSIT) and the
receivers (CSIR) \cite{Shamai-W-S}. However, even though the
assumption of perfect CSIR can be justified, it is unrealistic to
assume the same about CSIT. Moreover, the rate achievable over the
GBC is quite sensitive to the quality of CSIT as has been
demonstrated in \cite{Caire, Jindal, Vaze_Dof_final} (see also
references in \cite{Vaze_Dof_final}). This paper therefore tackles
the important problem of achieving high throughputs using DPC over
the GBC with imperfect CSIT.

It is well known that under perfect CSIT the DPC based transmission
outperforms other known strategies such as zeroforcing beamforming
(ZFBF) \cite{Caire}. Nevertheless, under {\em imperfect} CSIT, it is
the ZFBF strategy that has been intensely researched rather than DPC,
mainly because ZFBF is analytical tractable \cite{Jindal,Wei_GBC}
and because of a perception that DPC based schemes are either not feasible without perfect CSIT
or, even if feasible, they may be analytically intractable. The main hurdle with DPC is seen to be the
difficulty of designing the inflation factor -- a parameter that
can critically affect its performance \cite{Costa} -- without perfect CSIT;
it is even generally believed that the inflation factor cannot be effectively
designed without perfect CSIT \cite{Jindal} implying that DPC
may be overly sensitive to the imperfection in CSIT, thereby rendering it less desirable
than even ZFBF.

Making progress to this end, we recently developed iterative numerical algorithms for the
determination of inflation factor under imperfect CSIT which yield high achievable rates
\cite{Vaze,Vaze2}. Some analytical results were also obtained in the high/low SNR
(signal-to-noise ratio) regime \cite{Vaze_DPC_GBC,Vaze_fb_scaling_GBC}.  However,
these results may not always reveal much insight on how DPC works with imperfect CSIT
nor does it shed light on the behavior of DPC at moderate values of SNR.
Moreover, due to the numerical nature of these algorithms, it is almost
impossible to derive analytical results regarding the finite SNR performance of DPC based strategies
or on how they compare with other transmission strategies such as ZFBF. Furthermore,
the algorithms don't lend themselves to answering important design questions about DPC based schemes --
such as optimizing the sum-rate by selecting (and transmitting to) only a subset of users -- other than through a tedious and un-insightful exhaustive search. Recall that the strategy of transmitting to a subset of users
is known to indeed result in a considerable improvement in the sum-rate under
perfect CSIT \cite{Caire} and for the ZFBF even under imperfect CSIT \cite{Wei_GBC}.

To address the above issues, we undertake here a large-system or asymptotic analysis
of the GBC with $K$ transmit antennas and $K$ single-antenna users (i.e., the GBC of
size/dimension $K$) in which the CSIT is obtained via a finite-rate
feedback link of capacity $r$ bits per user per channel realization (or coherence interval).
In particular, for the transmission strategy of zeroforcing DPC (ZFDPC) \cite{Caire,
Vaze_DPC_GBC}, we analyze the normalized sum-rate or the throughput
(i.e., the sum-rate divided by $K$) of the GBC in the limit as $K
\to \infty$ with $\frac{r}{K} \to \bar{r}$. Such a problem has been
considered before in the special case of perfect CSIT (i.e.,
$\bar{r} = \infty$) in \cite{Caire} and for the finite-rate feedback
GBC with the simpler transmission strategy of ZFBF in \cite{Wei_GBC}.

\begin{figure*}[!b]
\begin{picture}(10,1)
\put(10,10){\line(1,0){500}}
\end{picture} \vspace{-14pt}
\begin{eqnarray}
R_i = \mathbb{E}_{\hat{H}} \max_{W_i} \mathbb{E}_{H|\hat{H}} \log
\frac{\mathrm{nr}}{\mathrm{nr} \cdot (1 + ||W_i||^2) - \frac{P}{s}
\Big| h_i^* \big(\omega_i + [\omega_1 \cdots \omega_{i-1}]W_i^*
\big)  \Big|^2}  \mbox{ where } \mathrm{nr} = 1 + \frac{P}{s}
\sum_{j=1}^{K} |h_i^* \omega_j|^2. \label{Achievable Rate R_i}\\
%{\left| \begin{array}{cc}
%1 + W_i W_i^* & \sqrt{\frac{P}{s}} \left(\omega_i^* + W_i [\omega_1 \cdots \omega_{i-1}]^* \right) h_i \ \\
%\sqrt{\frac{P}{s}} h_i^* \left(\omega_i + [\omega_1 \cdots \omega_{i-1}]W_i^* \right) & \mathrm{nr} \\
%\end{array} \right|},
W_i = \frac{P}{s} \mathbb{E}_{H|\hat{H}} \Big(\omega_i^* h_i h_i^*
[\omega_1 \cdots \omega_{i-1}] \Big) M_i^{-1} \mbox{ where } M_i =
\mathbb{E}_{H|\hat{H}} \Big(\mathrm{nr} \cdot I_{i-1} - \frac{P}{s}
[\omega_1 \cdots \omega_{i-1}]^* h_i h_i^* [\omega_1 \cdots
\omega_{i-1}] \Big). \label{Opt Choice W}
\end{eqnarray}
\end{figure*}

In the large-system limit, the involved random variables converge to
their deterministic limits \cite{TulinoVerdu}. Therefore, the
large-system analysis yields a closed-form expression for the
asymptotic throughput (i.e., the throughput in the limit of $K \to
\infty$), the evaluation of which involves a simple easy-to-compute
numerical integral. Importantly, unlike many works that deal with high SNR
characterizations (c.f., \cite{Vaze_Dof_final} and the references therein)
the asymptotic throughput is obtained as a function of SNR, and hence, it can
provide insights at any finite SNR. It also serves as a simple
semi-analytic tool for the comparison of different transmission
strategies. In particular, contrary to popular belief, we show that even under imperfect
CSIT ZFDPC does indeed achieve a significantly higher throughput than ZFBF.
Furthermore, the asymptotic analysis helps to definitively answer the design question
of optimizing over the number of users to be transmitted to. %Using the asymptotic throughput expression, we propose a simple method for solving this optimization.
It is also een that this method when mapped simply to finite dimensions
works quite accurately even for the relatively small values of $K$.
Thus the asymptotic analysis is seen to offer useful insights about
finite-dimensional GBCs as well.

\emph{\underline{Notations:} } For a matrix/vector $A$, $A^*$ is its
complex-conjugate transpose. $\mathcal{C}\mathcal{N}(0,1)$ denotes
the circularly symmetric complex normal mean-$0$ variance-$1$ random
variable (RV), while $\chi^2_{2K}$ denotes the chi-square RV of mean
$K$. $h \sim \mathcal{C}\mathcal{N}(K)$ denotes the vector $h$ of
dimension $K$ consisting of independent
$\mathcal{C}\mathcal{N}(0,1)$ RVs. For any vector $a$, $\tilde{a}$
denotes its direction, i.e., $\tilde{a} = \frac{a}{\|a\|}$, where $\| a \|$
denotes the norm of $a$. For a
vector $\hat{h}_i$, the perpendicular space of it and the
orthonormal basis vectors spanning the perpendicular space, both,
are denoted by $\hat{p}_i$ (the meaning is to be understood from the
context). Almost-sure convergence \cite{Pillai} is denoted by
a.s. $I_K \in \mathcal{C}^{K \times K}$ is an identity matrix.  For
RVs $A$ and $B$, $A \bot B$ denotes independence. All logarithms are
to base $2$.

\section{System Model of GBC} \label{SystemModel}
% \subsection{GBC}
Consider the GBC of size $K$. The received signal at the $i^{th}$
user is given by $y_i = h_i^* x + z_i$, where $h_i^* \in
\mathbb{C}^{1 \times K}$ is the channel vector of the $i^{th}$ user,
$x \in \mathbb{C}^{K \times 1}$ is the signal transmitted under the
power constraint of $P$, and $z_i \sim \mathcal{C}\mathcal{N}(0,1)$
is the additive noise. We assume that $h_i \sim
\mathcal{C}\mathcal{N}(K)$ are independent. Let $\omega_i \in
\mathbb{C}^{K \times 1}$ denote the BF vector for the $i^{th}$ user and let
$\|\omega_i\| = 1$. Let $u_i$ be the data symbol to be sent to the
user $i$ ($u_i$'s are independent). Then the total transmitted
signal is given by $x = \sum_{i=1}^K \omega_i u_i$. Let $H = [h_1
\hspace{5pt} h_2 \hspace{5pt} h_3 \cdots h_K]$. We assume perfect
CSIR. Define SNR $=P$.

W consider the so-called `on-off' power allocation policy which is that
 the transmitter selects a set, denoted $\mathcal{A}_{on}$,  of `on' users
and transmits with equal power to the selected users. In fact, under
perfect CSIT, such a scheme is near optimal. To be precise, the
difference between the asymptotic throughput achieved with the
optimal waterfilling-type power allocation policy \cite{Caire} and
that obtained using the on-off power policy is negligible.
Hence, we consider here only the on-off power policy.
Under this scheme, if $i \in
\mathcal{A}_{on}$ then $u_i \sim \mathcal{C}\mathcal{N}(0,
\frac{P}{s})$, where $s := |\mathcal{A}_{on}|$, else $u_i = 0$. We
let $\frac{s}{K} \to \bar{s}$ as $K \to \infty$. Thus, the
asymptotic throughput is obtained as a function of $\bar{s}$, which
allows us to answer the design problem of optimization over the
fraction of users.

%In the limit of large $K$, the normalized (by $K$) norm-squares of
%all the channel vectors converge to 1 in probability
%\cite[Proposition 1]{Wei_GBC}.

\subsection{Quantization Scheme}
In the limit of large $K$, RVs $\max_{1\leq i \leq K} \frac{1}{K}
||h_i||^2$ and $\min_{1\leq i \leq K} \frac{1}{K} ||h_i||^2$, both,
converge to $1$ in probability \cite[Proposition 1]{Wei_GBC}.
Therefore, we find it sufficient, for the present purpose, to
feedback only the channel directions, namely the $\tilde{h}_i$'s.
Let $r$ denote the number of feedback bits per user. Each user has a
codebook $\mathcal{C} = \{q_j\}_{j=1}^{2^r}$ consisting of $2^r$
$K$-dimensional unit-norm vectors. The vector $\tilde{h}_i$ is
quantized according to the rule: $ \hat{h}_i = \mathrm{arg}\min_{q_j
\in \mathcal{C}} \sin^2 \big (\angle(\tilde{h}_i,q_j)\big)$. Denote
by $d^2_c(i)$ (or simply $d_c^2$) the quantization error
$\sin^2\big(\angle(\tilde{h}_i,\hat{h}_i)\big)$. We further assume
that $\frac{r}{K} \to \bar{r}$. We define $\hat{H} = [\hat{h}_1
\hspace{5pt} \hat{h}_2 \hspace{5pt} \hat{h}_3 \cdots \hat{h}_K]$.

In this paper, we assume the quantization cell approximation (called
the quantization-cell upper-bound (QUB)) \cite{Mukkavilli}. Under
this approximation, we assume ideally that the quantization cell
around each vector of the codebook is a spherical cap of area
$2^{-r}$ times the total surface area of the unit sphere
\cite{Mukkavilli}. It has been shown that this approximation yields
an upper-bound to the performance \cite{Mukkavilli, YJ}; and this
upper-bound has been observed to be tight \cite{YJ}. The tightness
of the bound is a property that depends mainly on how the
quantization error is modeled under the approximation, and not so
much on the channel or the transmission scheme for which the
approximation is being used. Hence, QUB is a reasonable
assumption to make for the present analysis.

\begin{lemma} \label{lemma: Prop. of QUB}
If we write $\tilde{h}_i = \sqrt{1 - d_c^2(i)} \hat{h}_i +
\sqrt{d_c^2(i)} \tilde{e_i}$ then, under the QUB model, $\tilde{e}_i$ is
isotropically distributed in $\hat{p}_i$ and is independent of
$d_c^2(i)$. Also, $d_c^2(i) \to 2^{-\bar{r}} =: \bar{D}$ a.s.
\end{lemma}

We consider here an ensemble of codebooks $\{\bar{U}\mathcal{C}\}$
where $\bar{U}$ is a Haar distributed unitary matrix and $\mathcal{C}$
is a given codebook. Then $\hat{h}_i$ is isotropic. We take the
expectation over the codebooks as well although this is not shown
explicitly in subsequent formulas.

\subsection{The Achievable Rate}
Assume that the users are encoded according to their natural order.
We focus here on ZFDPC scheme \cite{Caire, Vaze_DPC_GBC}. To obtain
the BF vectors under this scheme, we perform QR-decomposition
\cite{Horn-Johnson} of $H$ when there is perfect CSIT \cite{Caire};
i.e., let $H = QR$. The columns of $Q$ are respectively the BF
vectors of the users. Under finite-rate feedback, we perform the
same procedure with $\hat{H}$ \cite{Vaze_DPC_GBC}. Note that the BF
vectors are orthogonal and (under finite-rate feedback) $\hat{h}_i^*
\omega_j = 0$, $\forall j>i$.

We select the auxiliary random variable for the $i^{th}$ user as
$U_i = u_i + W_i [u_1^* \cdots u_{i-1}^*]^*$, where $W_i \in
\mathcal{C}^{1 \times (i-1)}$ is the inflation factor for the
$i^{th}$ user \cite{Costa, Vaze2, Vaze_DPC_GBC}. Then the achievable
rate for the $i^{th}$ user is given by equation (\ref{Achievable
Rate R_i}) below (see \cite{Vaze2}).

In this scheme, the interference (at user $i$) due to the users
encoded previously (i.e., users $1$ to $i-1$) is treated by DPC
whereas that due to the users encoded afterwards (i.e., users $i+1$
to $s$) is treated by zeroforcing; hence the name ZFDPC.

Our choice of the inflation factor is stated in equation (\ref{Opt
Choice W}), which is obtained from \cite{Vaze2} and
\cite{Vaze_DPC_GBC}. It is derived by first moving the conditional
expectation in equation (\ref{Achievable Rate R_i}) inside the
logarithm to obtain an upper-bound on the rate and then maximizing
this upper-bound over the inflation factor.

\section{Evaluation of the Asymptotic Throughput}
We want to evaluate the limit, $\lim_K \frac{1}{K} \sum_i R_i$,
where each $R_i$ depends on the user index $i$, and hence, on the
normalized user index $\frac{i}{s}$. In the limit, the normalized
user index would take a value from the continuum, i.e., $\bar{i} :=
\lim_K \frac{i}{s} \in [0,1]$. We thus anticipate that, as $K \to
\infty$, the above summation would converge to an integral over
$\bar{i}$ and the integrand of which would be the limit of $R_i$.

To compute the limit, we first need to evaluate the conditional
expectations in (\ref{Opt Choice W}). We start below with
$\mathbb{E}_{H|\hat{H}} (\tilde{h}_i \tilde{h}_i^*)$. Then we
compute $\mathbb{E}_{H|\hat{H}} \mathrm{nr}$ and $\lim_K
\mathrm{nr}$. To this end, note that $\mathrm{nr} = 1 + \frac{P}{s}
|h_i^* \omega_i|^2 + \frac{P}{s} \sum_{j<i} |h_i^* \omega_j|^2 +
\frac{P}{s} \sum_{k>i} |h_i^* \omega_k|^2$; each of these terms is
dealt separately in Subsections \ref{SubSec: Term 1 of nr} to
\ref{subsec: 3rd term of nr}, respectively. Later, in Subsection
\ref{Compute_Wfirst}, we find $W_i$ in closed form. Finally,
we evaluate the limit of the terms involved in the denominator of the
argument of the logarithm in (\ref{Achievable Rate R_i}). We define
$D = \mathbb{E}_{H|\hat{H}} d_c^2(i)$. Then $\lim_K D = \bar{D} =
2^{-\bar{r}}$.

\begin{lemma} \label{lemma: Cond Mean}
$ \mathbb{E}_{H|\hat{H}} (\tilde{h}_i \tilde{h}_i^*) = (1-D -
\frac{D}{K-1}) \hat{h}\hat{h}^* + \frac{D}{K-1} I_K$.
\end{lemma}
\begin{proof}
We first prove that $\mathbb{E} \tilde{e}_i = 0$. To this end,
consider the unitary matrix $U = [\hat{h}_i \hspace{5pt}
\hat{p}_i]\left[ \begin{array}{cc}
1 & 0 \\
0 & V \\
\end{array} \right] \left[ \begin{array}{c}
\hat{h}_i^* \\
\hat{p}_i^* \\
\end{array} \right]$,
where $V \in \mathcal{C}^{(K-1) \times (K-1)}$ is any arbitrary
unitary matrix. Now, $U\hat{h}_i = \hat{h}_i$ and $U\hat{p}_i =
\hat{p}_i V$. Hence, $U \tilde{e}_i$ is isotropic in $\hat{p}_i$,
which implies that $U \mathbb{E}_{H|\hat{H}} (\tilde{e}_i) =
\mathbb{E}_{H|\hat{H}} (\tilde{e}_i) = 0$.

Now, $\mathbb{E}_{H|\hat{H}} (\tilde{e}_i \tilde{e}_i^*)$ must be of
the form $\hat{p}_i Q \hat{p}_i^*$ where $Q$ is positive
semi-definite matrix. We can prove that $U \mathbb{E}_{H|\hat{H}}
(\tilde{e}_i \tilde{e}_i^*) U^* = \mathbb{E}_{H|\hat{H}}
(\tilde{e}_i \tilde{e}_i^*) \Rightarrow VQV^* = Q$, $\forall$
unitary $V$. This implies that $Q = k \cdot I_{K-1}$ with $k$ chosen
such that $\mathrm{tr}(Q) = 1$. Hence $ \mathbb{E}_{H|\hat{H}}
(\tilde{e}_i \tilde{e}_i^*) = \hat{p}_i \frac{1}{K-1} I_{K-1}
\hat{p}_i^*$.

Now using the decomposition of $\tilde{h}_i$ given in Lemma
\ref{lemma: Prop. of QUB} and noting the fact that $\hat{p}_i
\hat{p}_i^* = I_K - \hat{h}_i \hat{h}_i^*$ we obtain the result.
\end{proof}

\subsection{Analysis of $\frac{1}{s} |h_i^* \omega_i|^2 $}
\label{SubSec: Term 1 of nr} \vspace{-5mm}
\begin{eqnarray*}
\lefteqn{\frac{1}{s} |h_i^* \omega_i|^2 =  \frac{1}{s} \left|h_i^*
\left[
\begin{array}{cc}
\hat{h}_i & \hat{p}_i \\
\end{array} \right] \left[ \begin{array}{cc}
\hat{h}_i & \hat{p}_i \\
\end{array} \right]^* \omega_i \right|^2 }\\
&& {} = \frac{1}{s} \Big\{ \Big| h_i^* \hat{h}_i \hat{h}_i^*
\omega_i \Big|^2 + \Big| h_i^* \hat{p}_i \hat{p}_i^* \omega_i
\Big|^2 + \mbox{ cross terms} \Big\}.
\end{eqnarray*}
1) $\frac{1}{s} |h_i^* \hat{h}_i \hat{h}_i^* \omega_i|^2 =
\frac{||h_i||^2}{s} |\tilde{h}_i^* \hat{h}_i|^2 |\hat{h}_i
\omega_i|^2$. Therefore, $ \mathbb{E}_{H|\hat{H}} \frac{1}{s} |h_i^*
\hat{h}_i \hat{h}_i^* \omega_i|^2 = \frac{K}{s} (1-D) |\hat{h}_i
\omega_i|^2$.

We know $|\tilde{h}_i^* \hat{h}_i|^2 \to 1-\bar{D}$ a.s. Note that
$\Big| ||h_i|| \cdot \hat{h}^*_i \omega_i \Big| ^2 \sim
\chi^2_{2(K-i+1)}$ \cite{TulinoVerdu}. Hence, $\Big|\hat{h}^*_i
\omega_i \Big| ^2 = \frac{s}{||h_i||^2} \frac{1}{s} \Big| ||h_i||
\hat{h}^*_i \omega_i \Big| ^2 \to (1 - \bar{i}\bar{s})$ a.s.
Therefore $\frac{1}{s} |h_i^* \hat{h}_i \hat{h}_i^* \omega_i|^2 \to
\frac{1}{\bar{s}} (1 - \bar{D}) (1 - \bar{i}\bar{s})$ a.s.

2) Now $\hat{p}_i \hat{p}_i^* \tilde{h}_i = e_i$ with
$\mathbb{E}_{H|\hat{H}} ||e_i||^2 = D$. Let us consider $h' \sim
\mathcal{C}\mathcal{N}(K-1)$ $\bot \tilde{e}_i$. Let $a_i =
\hat{p}_i \hat{p}_i^* \omega_i$; $||a_i||^2 = 1 - |\hat{h}_i^*
\omega_i|^2$. Conditioned on $\hat{H}$, $\tilde{a}_i \in \hat{p}_i$
is a deterministic direction. Hence, conditioned on $\hat{H}$,
$\tilde{a}_i^* \cdot ||h'|| \tilde{e}_i \sim
\mathcal{C}\mathcal{N}(0,1)$. Therefore, \newline
$\mathbb{E}_{H|\hat{H}} \Big|\tilde{e}_i^* \omega_i \Big|^2  =
\frac{||a_i||^2}{K-1} \mathbb{E}_{H|\hat{H}} \Big| \tilde{a}_i^*
\cdot ||h'|| \tilde{e}_i \Big|^2  = \frac{1 - |\hat{h}_i^*
\omega_i|^2}{K-1}$ and $\mathbb{E}_{H|\hat{H}} \frac{1}{s}|h_i^*
\hat{p}_i \hat{p}_i^* \omega_i |^2 = \frac{D K}{s(K-1)}(1 -
|\hat{h}_i^* \omega_i|^2)$.

To compute the limit, note that $\tilde{a}_i^* \tilde{e}_i$ behaves
like $\frac{\mathcal{C}\mathcal{N}(0,1)}{K-1} $, as far as the limit
is concerned. Since the other multiplicative terms remain bounded in
limit, this term converges to zero a.s.

3) One of the cross terms is $||h_i||^2 \hspace{2pt} \tilde{h}_i^*
\hat{h}_i \hat{h}_i^* \omega_i \omega_i^* \hat{p}_i \hat{p}_i^*
\tilde{h}_i$. Now, $\tilde{h}_i^* \hat{h}_i = \sqrt{1 - d_c^2(i)}$;
$\hat{p}_i \hat{p}_i^* \tilde{h}_i = \sqrt{d_c^2(i)} \tilde{e}_i$.
Since $d_c(i) \bot \tilde{e}_i$ and $\mathbb{E}_{H|\hat{H}}
\tilde{e}_i = 0$, the conditional expectation of the cross terms is
zero. Their limit can also be shown to be zero a.s.

\subsection{Analysis of $\frac{1}{s}\sum_{j<i} |h_i^*
\omega_j|^2$}

The conditional expectation can be computed using the techniques
developed in Subsection \ref{SubSec: Term 1 of nr}. We directly
state the main result. Note that $|\hat{h}_i^* \omega_i|^2 = 1 -
\sum_{j<i} |\hat{h}_i^* \omega_j|^2$ and \newline $ |h_i^*
\omega_j|^2 = |h_i^* \hat{h}_i \hat{h}_i^* \omega_j|^2 + |h_i^*
\hat{p}_i \hat{p}_i^* \omega_j |^2 + \mbox{ cross terms}$. Then
\vspace{3pt} \newline $\mathbb{E}_{H|\hat{H}} \sum_{j<i} |h_i^*
\hat{h}_i \hat{h}_i^* \omega_j|^2 = K(1-D) (1 - |\hat{h}_i^*
\omega_i|^2)$. \vspace{2pt}
\newline $\mathbb{E}_{H|\hat{H}} \sum_{j<i} |h_i^*
\hat{p}_i \hat{p}_i^* \omega_j |^2 = \frac{K D}{(K-1)} \Big(i-2 +
|\hat{h}_i^*\omega_i|^2 \Big)$. \vspace{4pt}

To compute the limit, we have $\omega_j \bot h_i$ $\forall j<i$.
Thus, $h_i^* \omega_j$ are independent $\sim
\mathcal{C}\mathcal{N}(0,1)$ random variables. Hence
$\frac{1}{s}\sum_{j<i} |h_i^* \omega_j|^2 \to \bar{i}$ a.s.

\subsection{Analysis of $\frac{1}{s} \sum_{k>i} | h_i^*
\omega_k|^2$} \label{subsec: 3rd term of nr}

To compute the conditional expectation, the techniques developed in
Subsection \ref{SubSec: Term 1 of nr} are used. We omit the details
and state the main result: $ \frac{1}{s} \mathbb{E}_{H|\hat{H}}
\sum_{k>i} | h_i^* \omega_k|^2 = \frac{K}{K-1} D \frac{s-i}{s}.$

Since $\omega_k \in \hat{p}_i$, we get $| h_i^* \omega_k|^2 =
||h_i||^2 \cdot ||e_i||^2 \cdot  |\omega_k^* \tilde{e}_i|^2$. As in
Subsection \ref{SubSec: Term 1 of nr}, introduce $h' \sim
\mathcal{C}\mathcal{N}(K-1) \bot \tilde{e}_i$. Then conditioned on
$\hat{H}$, $\sum_{k>i} \Big| \omega_k^* \cdot ||h'|| \tilde{e}_i
\Big| \sim \chi^2_{2(s-i)}$; hence unconditionally also, it has the
same distribution. This gives $\frac{1}{s} \sum_{k>i} | h_i^*
\omega_k|^2 \to \bar{D} (1-\bar{i})$ a.s.

\begin{figure*}[!b]
\begin{picture}(10,1)
\put(10,10){\line(1,0){500}}
\end{picture} \vspace{-12pt}
\begin{eqnarray}
\rho (P,\bar{s}, \bar{r}) = \bar{s} \left\{ \log \mathrm{NR} -
\int_0^1 \log \left( \mathrm{NR} \cdot x^{\infty}(\bar{i}) -
\frac{P}{\bar{s}}(1-2^{-\bar{r}}) (1-\bar{i}\bar{s})
\Big(\frac{\frac{P}{\bar{s}} (1-2^{-\bar{r}}) (\bar{i}\bar{s})}{1 +
\frac{P}{\bar{s}}(1-2^{-\bar{r}}) (1-\bar{i}\bar{s}) +
P2^{-\bar{r}}} + 1 \Big)^2 \right)  \,d\bar{i} \right\},
\label{Asymp. Through. Wopt} \\
&& {} \hspace{-18.2cm} \hspace{-1pt} \mbox{where } \mathrm{NR}
\hspace{-1pt} := \hspace{-2pt} \lim_{K \to \infty} \hspace{-1pt}
\mathrm{nr} = 1 \hspace{-1pt} + \hspace{-1pt} \frac{P}{\bar{s}}(1-
2^{-\bar{r}}) + P 2^{-\bar{r}} \mbox{ is independent of $\bar{i}$;
and } x^{\infty}(\bar{i}) = 1 + \frac{(\frac{P}{\bar{s}})^2
(1-2^{-\bar{r}})^2 (1 - \bar{i}\bar{s}) (\bar{i}\bar{s}) }{\left(1 +
P 2^{-\bar{r}} + \frac{P}{\bar{s}} (1 -2^{-\bar{r}}) (1-
\bar{i}\bar{s}) \right)^2}. \nonumber
\end{eqnarray}
\end{figure*}

%We first compute $ \mathbb{E}_{H|\hat{H}} ( \omega_i^* \tilde{h}_i
%\tilde{h}_i^* [\omega_1 \cdots \omega_{i-1}] )$,
%$\mathbb{E}_{H|\hat{H}} \mathrm{nr}$, and $M_i$.

\subsection{Computation of the inflation factor}
\label{Compute_Wfirst} Let us define $l_i^* = \hat{h}_i^* [\omega_1
\hspace{5pt} \omega_2 \cdots \omega_{i-1}]$. Then we have:
\vspace{3pt}
\newline $\mathbb{E}_{H|\hat{H}} ( \omega_i^* \tilde{h}_i
\tilde{h}_i^* [\omega_1 \cdots \omega_{i-1}] )= (1-D -\frac{D}{K-1})
(\omega_i^* \hat{h}_i) l_i^*$, \newline $\mathbb{E}_{H|\hat{H}}
\mathrm{nr} = \overline{\mathrm{nr}} = 1 + \frac{PK}{s}(1-D) + PD
\frac{K}{K-1}\frac{s-1}{s}$, and \newline $ M_i =
\left(\overline{\mathrm{nr}} - \frac{PDK}{s(K-1)} \right)I_{i-1} -
\frac{PK}{s} \left(1-D - \frac{D}{K-1} \right) l_i l_i^*.$

We now need $l_i^* M_i^{-1}$. Note that $M_i$ is positive definite
and one of the eigenvectors of $M_i$ or $M_i^{-1}$ is $\tilde{l}_i$
while the rest are orthogonal to $\tilde{l}_i$. Hence, $l_i^*
M_i^{-1}$ is equal to $l_i^*$ times the eigenvalue of $M_i^{-1}$
corresponding to $\tilde{l}_i$ as the eigenvector. Therefore, we get
\begin{equation}
W_i = \frac{\frac{PK}{s} \Big( 1-D -\frac{D}{K-1} \Big)
(\omega_i^*\hat{h}_i)l_i^*}{1 + \frac{PK}{s} \Big(1-D -\frac{D}{K-1}
\Big) |\hat{h}^* \omega_i|^2 + P D \frac{K}{K-1} }. \label{Optimal
Choice W_i closed-form}
\end{equation}
%Having obtained $W_i$ in closed form, we can now evaluate the limit
%of $(1+ ||W_i||^2)$, which we denote by $x^{\infty}(\bar{i})$. To
%this end, note that $||l_i||^2 = \sum_{j<i} |\hat{h}_i^* \omega_j|^2
%= 1 - |\hat{h}_i^* \omega_i|^2$ and $|\hat{h}_i^* \omega_i|^2 \to (1
%- \bar{i}\bar{s})$ a.s. Then we have
%\begin{eqnarray*}
%x^{\infty}(\bar{i}) = 1 + \frac{(\frac{P}{\bar{s}})^2 (1-\bar{D})^2
%(1 - \bar{i}\bar{s}) (\bar{i}\bar{s}) }{\left(1 + P\bar{D} +
%\frac{P}{\bar{s}} (1 - \bar{D}) (1- \bar{i}\bar{s}) \right)^2}
%\mbox{ a.s.}
%\end{eqnarray*}

\subsection{Analysis of $(1 + ||W_i||^2)$}
We have $||l_i||^2 = \sum_{j<i} |\hat{h}_i^* \omega_j|^2 = 1 -
|\hat{h}_i^* \omega_i|^2$. Further $|\hat{h}_i^* \omega_i|^2 \to (1
- \bar{i}\bar{s})$ a.s. Let  $ x^{\infty}(\bar{i})$ denote the limit
of $(1 + ||W_i||^2)$. Then we can easily obtain the following:
\begin{eqnarray*}
x^{\infty}(\bar{i}) = 1 + \frac{(\frac{P}{\bar{s}})^2 (1-\bar{D})^2
(1 - \bar{i}\bar{s}) (\bar{i}\bar{s}) }{\left(1 + P\bar{D} +
\frac{P}{\bar{s}} (1 - \bar{D}) (1- \bar{i}\bar{s}) \right)^2}
\mbox{ a.s.}
\end{eqnarray*}

\subsection{Analysis of $f_i := \frac{1}{s} \Big|h_i^* (\omega_i + [\omega_1 \cdots \omega_{i-1}]W_i^*)\Big|^2$}
From equation (\ref{Optimal Choice W_i closed-form}), we can write
$W_i^* = (\hat{h}_i^* \omega_i) c_i l_i$ for an appropriately chosen
scalar $c_i$. Then
\begin{eqnarray*}
f_i & = & \hspace{-7pt} \frac{1}{s} \Big|h_i^*[\omega_1 \cdots
\omega_{i-1}] (\hat{h}_i^* \omega_i) c_i l_i + h_i^* \omega_i \Big|^2 \nonumber \\
& = & \hspace{-7pt} \frac{||h_i||^2}{s|\hat{h}_i^* \omega_i|^2}
\Big| \tilde{h}_i^* [ \omega_1 \cdots \omega_{i-1}] l_i c_i
|\hat{h}_i^* \omega_i|^2 + \tilde{h}_i^* \omega_i \omega_i^*
\hat{h}_i \Big|^2 \hspace{-2pt} . \label{eq: subsecF main term}
\end{eqnarray*}
Let us first analyze the terms $A = \tilde{h}_i^* [ \omega_1 \cdots
\omega_{i-1} ] l_i$ and $B = \tilde{h}_i^* \omega_i \omega_i^*
\hat{h}_i $ in the above equation. To this end, we have
\begin{eqnarray*}
\lefteqn{ \hspace{8pt} A = \tilde{h}_i^* \big[\hat{h}_i \hspace{7pt}
\hat{p}_i \big] \hspace{4pt} \big[\hat{h}_i \hspace{7pt} \hat{p}_i
\big]^* \hspace{4pt} \big[ \omega_1 \cdots \omega_{i-1} \big] l_i } \nonumber \\
&& {} \hspace{-5pt}  =  \tilde{h}_i^* \hat{h}_i \hat{h}_i^*
\hspace{2pt} \big[\omega_1 \cdots \omega_{i-1} \big] l_i +
\tilde{h}_i^* \hat{p}_i \hat{p}_i^* \hspace{2pt} \big[
\omega_1 \cdots \omega_{i-1} \big] l_i, \\
&& {} \hspace{-32pt} \mbox{and } B = \tilde{h}_i^* \hat{h}_i
\hat{h}_i^* \omega_i \omega_i^* \hat{h}_i + \tilde{h}_i^* \hat{p}_i
\hat{p}_i^* \omega_i \omega_i^* \hat{h}_i.
\end{eqnarray*}
Using the arguments developed in Subsection \ref{SubSec: Term 1 of
nr} Part 2), it can be proved that the terms $\tilde{h}_i^*
\hat{p}_i \hat{p}_i^* \hspace{2pt} \big[ \omega_1 \cdots
\omega_{i-1} \big] l_i$ and $\tilde{h}_i^* \hat{p}_i \hat{p}_i^*
\omega_i \omega_i^* \hat{h}_i$ converge to zero a.s. Hence, we get
\begin{eqnarray*}
\lefteqn{ \hspace{-0cm} \lim_K \Big|A c_i |\hat{h}_i^* \omega_i|^2 + B \Big|^2} \\
& = &  \lim_K \Big| \tilde{h}_i^* \hat{h}_i
\hat{h}_i^* [\omega_1 \cdots \omega_{i-1}] l_i c_i |\hat{h}_i^*
\omega_i|^2 + \tilde{h}_i^* \hat{h}_i \hat{h}_i^*
\omega_i \omega_i^* \hat{h}_i \Big|^2 \\
& = & \lim_K |\tilde{h}_i^* \hat{h}_i|^2 \cdot |\hat{h}_i^*
\omega_i|^4 \cdot \Big| c_i ||l_i||^2 + 1 \Big|^2.
\end{eqnarray*}

Now, $ |\tilde{h}_i^* \hat{h}_i|^2 \to (1 - \bar{D})$ a.s.,
$|\hat{h}_i^* \omega_i|^4 \to (1-\bar{i}\bar{s})^2$ a. s., and
$||l_i||^2 \to (\bar{i}\bar{s})$. The limiting value of $c_i$,
denoted by $c^{\infty} (\bar{i})$, can be easily obtained from the
results of the previous subsection (not shown here explicitly).
Putting together these results, we obtain $ f_i \to
\frac{1}{\bar{s}} (1-\bar{D})(1-\bar{i}\bar{s}) \big(c^{\infty}
({\bar{i}}) \cdot \bar{i}\bar{s} + 1 \big)^2 \mbox{ a.s.}$

Using all the above results, we obtain the asymptotic throughput
$\rho(P,\bar{s},\bar{r})$ as given by equation (\ref{Asymp. Through.
Wopt}).

In the simpler case of perfect CSIT (i.e., as $\bar{r} \to \infty$),
the expression for the asymptotic throughput obtained here reduces
to $\rho (P,\bar{s},\infty) =  \bar{s} \int_0^1 \log \left\{1 +
\frac{1}{\bar{s}}P(1-\bar{i}\bar{s}) \right\} \,d\bar{i}$, which is
same as what one would obtain by specializing the result of
\cite{Caire} to the `on-off'-type power policy.

\section{Numerical Results}
Let the optimal value of $\bar{s}$ maximizing $\rho$ for the given
value of $P$ and $\bar{r}$ be $\bar{s}_{opt}$. Let $\rho_{opt}
(P,\bar{r}) := \rho(P,\bar{s}_{opt},\bar{r})$. The asymptotic
throughput for ZFBF is obtained from \cite{Wei_GBC}.

\begin{figure}
\centering
\includegraphics[height=2.2in,width=3.3in]{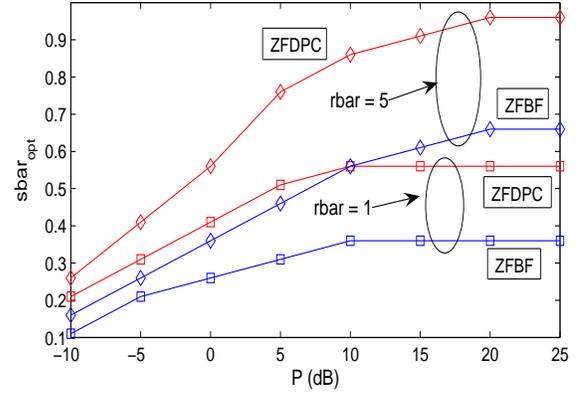}
\caption{$\bar{s}_{opt}$ vs. P for $\bar{r} = 1$ and $5$.}
\label{sbar_opt_VS_P_fixrbar}
\end{figure}

In Fig. \ref{sbar_opt_VS_P_fixrbar}, we plot $\bar{s}_{opt}$ as
function of $P$ for $\bar{r} = 1$, $5$. It can be seen that for
lower values of $P$, $\bar{s}_{opt}$ increases with $P$. This
behavior can be understood by noting that if the increased power is
allocated to only a few users, then there would be a `logarithmic'
increase in the sum-rate; however, if the power is distributed
across more users then the `pre-log' factor gets improved (i.e.,
more users contribute to the sum-rate). However, increasing
$\bar{s}$ also increases the inter-user interference. Hence, at
higher values of $P$, $\bar{s}_{opt}$ becomes constant. Next,
$\bar{s}_{opt}$ increases with $\bar{r}$ because the inter-user
interference reduces with increasing $\bar{r}$. Finally, note that
$\bar{s}_{opt}$ for ZFDPC is higher than that for ZFBF. As we would
discuss below, ZFDPC manages the channel gain of the useful signal
and the inter-user interference more efficiently than ZFBF and
hence, $\bar{s}_{opt}$ corresponding to it turns out to be higher.

\begin{figure}
\centering
\includegraphics[height=2.2in,width=3.2in]{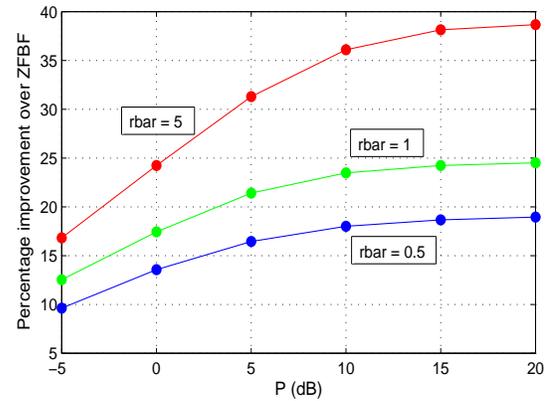}
\caption{Percentage improvement in the asymptotic throughput
achieved with ZFDPC over ZFBF.} \label{ImprovementOver_ZFBF}
\end{figure}

In Fig. \ref{ImprovementOver_ZFBF}, we compare the asymptotic
throughput obtained using ZFDPC and ZFBF. The numerical results in
this figure pertain to $\rho_{opt}$. Here we plot the percentage
improvement achieved using ZFDPC over ZFBF against $P$ for three
values of $\bar{r}$. We can see that ZFDPC achieves a considerably
higher throughput than ZFBF at all values of $P$ and $\bar{r}$. Note
that even for $\bar{r}$ as low as $0.5$, the percentage improvement is in
the range of $10 \%$ to $20\%$, which is significant.

\begin{figure}
\centering
\includegraphics[height=2.1in,width=3.25in]{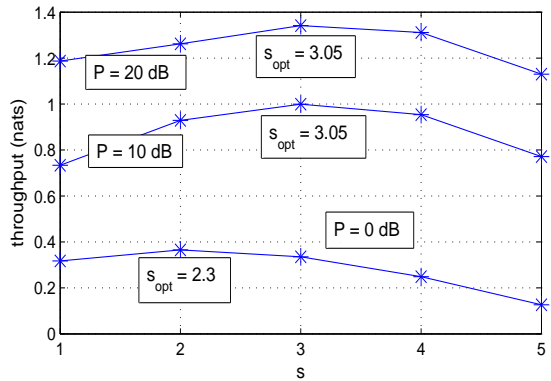}
\caption{Throughput achieved using ZFDPC vs. $P$ for $K=5$ and $r =
10$.} \label{rbar2_L5_rho_new}
\end{figure}

Let us now examine the differences between ZFDPC and ZFBF. Consider
the $i^{th}$ user. Under ZFDPC, the interference due to users $1$ to
$i-1$ is canceled by DPC. This can be accomplished for any given
choice of the BF vectors of users $1$ to $i$, as long as the
inflation factor $W_i$ is chosen in the appropriate manner. The BF
vectors are chosen under ZFDPC in such a way that the interference
due to users $i+1$ to $s$ gets zeroforced at the $i^{th}$ user. In
contrast to this, under ZFBF, the BF vectors are selected so as to
zeroforce the interference due to all other users. With this
background, let us now analyze the channel gain of the useful
signal, i.e., the term $\frac{1}{s} |h_i^* \omega_i|^2$. Under
ZFDPC, it is proportional the $\frac{1}{s} \chi^2_{2(K-i+1)}$ RV
(recall that the other additive terms converge to zero in limit),
whereas under ZFBF, it is proportional to $\frac{1}{s}
\chi^2_{2(K-s+1)}$ RV $\forall i$ \cite{Wei_GBC}. Thus, except for
the user $i=s$, every other user experiences a stronger channel
under ZFDPC. In other words, DPC (or ZFDPC) manages
the channel gain of the useful signal and the interference together
more effectively than ZFBF. It was known that, due to these
differences, ZFDPC outperforms ZFBF under perfect CSIT \cite{Caire}.
In the light of the results obtained here, we conclude that the same
is true even under imperfect CSIT as well. Lastly, it must be noted that
with ZFBF, the asymptotic throughput is zero when $\bar{s}=1$, i.e.,
$\rho(P,1,\bar{r}) =0$ $\forall P$, $\bar{r}$. Note that for ZFDPC,
$\rho(P,1,\bar{r})$ is comparable to $\rho_{opt}(P,\bar{r})$. This
behavior can be easily understood by noting the distribution of the
channel gain of the useful signal under two transmission schemes.

Consider now Fig. \ref{rbar2_L5_rho_new}. Here, we plot the
throughput (i.e., the sum-rate normalized by $K$) for the GBC with
$K=5$ and $r = 10$. We see from the figure that by optimizing over
the number of users, an improvement of about $0.2$ nats can be
obtained (at all values of $P$), over the simple solution of
transmitting to all $K$ users. This improvement is quite
significant, especially at $P = 0$ dB and $P = 10$ dB. The question
now is how to determine the optimal number of `on' users in the
$K$-dimensional GBC for a given $P$ and $r$. Instead of performing
an exhaustive search over $s$, we propose a simple and
computationally efficient approach. First find $\bar{s}_{opt}$ for
the given $P$ and $\bar{r} = \frac{r}{K}$. Then we suggest that for
the $K$-dimensional GBC, select $s_{opt} := \bar{s}_{opt} K$
(rounded to the nearest integer) number of users. In Fig.
\ref{rbar2_L5_rho_new}, we see that the maximum value of the
normalized throughput is indeed attained at $s_{opt}$ (rounded to
the nearest integer). We have observed this simple method to work
quite accurately, even for the relatively small values of $K$ ($K=5$
here). Note that the method suggested for the ZFBF in \cite{Wei_GBC}
using their large system analysis
for selecting the number of users is more complicated than the one
suggested here but seems to provide no particular benefit over
this simple approach.

\section{Conclusion}
We provide a large-system analysis of the GBC with finite-rate
feedback and derive a closed-form expression for the asymptotic
throughput achievable using ZFDPC. Using this result, we show that
the DPC-based scheme achieves a significantly higher throughput than
ZFBF. For the first time, DPC is shown to have a better performance,
under imperfect CSIT. Also, using the asymptotic throughput
expression, we address the problem of optimizing over the number of
`on' users.

% *******REFERENCES*********************
\bibliographystyle{IEEEtran}
\bibliography{Revisedv2_LargeSystem_BC}
\end{document}